\documentclass[10pt]{article}

\usepackage{a4wide}

\usepackage{amsfonts,epsfig,rotating,amssymb,framed}

\newtheorem{theorem}{Theorem}[section]

\newtheorem{corollary}[theorem]{Corollary}

\newtheorem{definition}[theorem]{Definition}
\newtheorem{remark}[theorem]{Remark}
\newtheorem{construction}[theorem]{Construction}

\newenvironment{proof}{\noindent{\bf Proof~}}{\null\hfill $\Box$\par\medskip}

\usepackage{mathrsfs}
\usepackage{colonequals}
\usepackage{paralist}
\usepackage{amsmath}
\usepackage{listings}
\usepackage{subfig}
\lstdefinelanguage{ mylng }{
	morekeywords={while,true}
}
\newcommand{\q}[1]{``#1''}
\newcommand{\fu}[1]{\mathcal{#1}}
\newcommand{\mc}[1]{\mathsf{#1}}
\newcommand{\ri}[1]{\mathscr{#1}}
\newcommand{\co}[1]{\fu{C}({#1})}
\newcommand{\lc}[1]{\alpha(#1)}
\newcommand{\rc}[1]{\beta(#1)}
\newcommand{\ema}[1]{\mathcal{#1}}
\newcommand{\fe}[2]{\fu{F}(#1,#2)}
\newcommand{\lv}[1]{l_{#1}}
\newcommand{\rv}[1]{r_{#1}}
\newcommand{\mv}[1]{m_{#1}}
\newcommand{\lvv}{\lv{v}}
\newcommand{\rvv}{\rv{v}}
\newcommand{\mvv}{\mv{v}}
\newcommand{\svv}{s_v}
\newcommand{\cordered}{\fu{S}}
\newcommand{\lvi}[2]{\lv{}^{#2}(#1)}%
\newcommand{\rvi}[2]{\rv{}^{#2}(#1)}%
\newcommand{\lvd}[1]{\fu{L}_v(#1)}
\newcommand{\rvd}[1]{\fu{R}_v(#1)}
\newcommand{\dom}{\fu{D}(G)}
\newcommand{\first}[1]{\fu{H}(#1)}
\newcommand{\last}[1]{\fu{T}(#1)}

\makeatletter
\renewcommand{\p@enumii}{\theenumi.(}

\newcommand{\ssref}[1]{(\ref{#1})}
\makeatother

\begin{document}

\title{Interval Routing Schemes for Circular-Arc Graphs}

\author{
Frank Gurski \\
University of D\"usseldorf\\
Institute of Computer Science\\
Algorithmics for Hard Problems Group\\
D-40225 D\"usseldorf\\
\texttt{frank.gurski@hhu.de}
\and 
Patrick Gwydion Poullie\thanks{Corresponding author.} \\
University of Zurich \\
Department of Informatics \\ Communication Systems Group\\
CH-8050 Z\"urich\\
\texttt{poullie@ifi.uzh.ch}
}

\maketitle


\begin{abstract}
Interval routing is a space efficient method to realize a distributed 
routing function.
In this paper we show that every circular-arc graph allows a shortest 
path strict 2-interval routing scheme, i.e., by introducing a global order
on the vertices and assigning at most two (strict) intervals in this order
to the ends of every edge allows to depict a routing function that 
implies exclusively shortest paths. 
Since circular-arc graphs do not allow shortest path 1-interval routing 
schemes in general, the result implies that the class of circular-arc 
graphs has strict compactness 2, which was a hitherto open question.
Additionally, we show that the constructed 2-interval routing scheme is 
a 1-interval routing scheme with at most one additional 
interval assigned at each vertex and we outline an algorithm 
to calculate the routing scheme for circular-arc graphs in 
$\mathcal{O}(n^2)$ time, where $n$ is the number of vertices.

\end{abstract}

\section{Introduction}

Routing is an essential task that a network of processors or computers must be able to 
perform. Interval routing is a space-efficient solution to this problem.
Sets of consecutive destination addresses that use the same output port are grouped into 
intervals and then assigned to this port.
In this way, the storage space required is greatly reduced compared to the straight 
forward approach, in which an output port is stored specifically for every destination address.
Of course, the advantage in space efficiency depends heavily on the number of intervals assigned 
to the output ports, which, in turn, depends on the address and network topology.
Interval routing was first introduced in \cite{Santoro82routingwithout,Santoro85labelling}; 
for a discussion of interval routing we refer the reader 
to \cite{Leeuwen87interval,FredericksonJ88,Bakker91linearinterval} 
and for a detailed survey to \cite{Gavoille00asurvey}.
When theoretical aspects of interval routing are discussed, the network is represented 
by a directed (symmetric) graph $G=(V,A)$. 
As with many routing methods, the problem is that space efficiency and path optimality are 
conflicting goals \cite{65953}: as shown in \cite{Guevremont98worstcase}, for every $n\in\mathbb{N}$, 
there exists an $n$-vertex graph $G_n$ such that for every shortest 
path interval routing scheme (for short IRS) for $G_n$ the maximal number of intervals per directed edge 
is only bounded by $\Omega(n)$. Also, it is NP-hard to determine the most space efficient 
shortest path IRS for a given graph \cite{Eilam2002,Flammini1997}.
Other worst case results can be found in \cite{savio99onthespace}.
On the other hand, there are many special graph classes \cite{FG98} including 
random graphs \cite{GP1998} that are known to allow shortest path IRSs 
with a constantly bounded number of intervals on all directed edges. 
If this number is a tight bound, it is denoted the compactness of the graph.
The compactness of undirected graphs is defined as the compactness of their directed 
symmetric version\footnote{The \emph{directed symmetric version of an undirected graph} 
$G_u=(V,E)$ is the directed graph $G_d=(V,A)$ with $A=\{(v,w),(w,v)~|~\{v,w\}\in E\}.$}
and the compactness of a graph class as the smallest $k$, such that every graph in this 
class has compactness  at most $k$, if such $k$ exists.

In this paper we show that the class of circular-arc graphs has (strict) compactness 2 
by presenting an algorithm to construct a corresponding IRS. A different approach to realize 
space efficient, shortest path routing in circular-arc graphs can be found 
in \cite{f.dragan:new}. In \cite{GP08} an $\mathcal{O}(\log(|V|))$-bit distance
labeling scheme is developed that allows for each pair of vertices in circular-arc graphs to
compute their exact distance in $\mathcal{O}(1)$ time.
In \cite{DYL06} an $\mathcal{O}(\log^2(|V|))$-bit routing labeling scheme is developed that allows 
to make a routing decision in  $\mathcal{O}(1)$ time for every vertex in an arbitrary circular-arc graph. 
The resulting routing path has at most two more edges than the respective shortest path.
Nevertheless, our result is interesting, since interval graphs and 
unit circular-arc graphs  are included in the class of circular-arc graphs and known to have 
compactness 1 \cite{FG98,NS96}, while for the compactness of circular-arc graphs hitherto 
only a lower bound of 2 was known (cf. the graph in Fig. \ref{fig:animals}) 
and the solutions from  \cite{FG98} and \cite{NS96} 
cannot be extended to the class of circular-arc graphs.
   
While interval graphs can be represented by intervals on a line, circular-arc graphs can be 
represented by arcs on a circle, that is to say, every circular-arc graph is the intersection 
graph of a set of arcs on a circle.
The first polynomial-time ($\ema{O}(|V|^3)$) algorithm 
to recognize circular-arc graphs and give a corresponding  circular-arc model can be found 
in \cite{firstCAG}. An algorithm with linear runtime for the same purpose is given 
in \cite{MCC03}. In circular-arc graphs maximum cliques can 
be computed much faster than in arbitrary graphs; an algorithm that determines a maximum 
clique in a circular-arc graph in $\ema{O}(|E|+|V|\cdot\log(|V|))$ time or even in $\ema{O}(|E|)$ time, 
if the circular-arc endpoints are given in sorted order, is presented 
in \cite{Bhattacharya1997336}. A discussion of circular-arc graphs can be found 
in \cite{Hsu95}. Circular-arc graphs have, amongst others, applications 
in cyclic scheduling, compiler design \cite{Gol04,tucker:493}, 
and genetics \cite{Rob76}.

This paper is organized as follows.
In Section \ref{chap:pre} we define cyclic permutations, circular-arc graphs, 
and  interval routing schemes.  
In Section \ref{main} we prove 
the main result of this paper, namely that every circular-arc graph allows a shortest path 
strict 2-interval routing scheme. 
To this end, we choose an arbitrary circular-arc graph $G=(V,A)$ and 
fix a cyclic permutation $\ri{L}$  on $V$, in Section \ref{constrCO}. 
In Section \ref{furtherDefs} we fix an arbitrary vertex $v$ and partition $V-\{v\}$ into the three sets $A_v$, $B_v$, and $C_v$; the vertices in these sets appear consecutively in $\ri{L}$ and are referred to as  the \emph{vertices to the right of $v$}, the \emph{vertices face-to-face with $v$}, and the \emph{vertices to the left of $v$}, respectively.
In Section \ref{2-app} we show that the vertices in the 
three sets  $A_v$, $B_v$, and $C_v$ can be assigned to the edges incident to $v$, such that 
the conditions of a shortest path 
strict 2-interval routing scheme are fulfilled. 
In Section \ref{improving} we show that the upper bound of the number of intervals in a shortest path strict 2-interval routing scheme constructed in this way
is close to the number in an 1-IRS.
In Section \ref{implement} we outline  how to 
implement the implied IRS in $\mathcal{O}(|V|^2)$ time.

\section{Preliminaries\label{chap:pre}}

\subsection{Cyclic Permutations}\label{cyclicOrders}

Cyclic permutations are relevant for the two main subjects of this paper, 
namely circular-arc graphs and non-linear interval routing schemes. 
Just like a permutation, a cyclic permutation implies an order on a set 
with the only difference being that there is no distinct first, second, $\ldots$,
or $n$-th element. We also define cyclic permutations on multisets, where
we assume that multiple appearances of an element are
distinguishable, i.e., 
a multiset $A=\{a_1,\ldots,a_n\}$ on $n$ elements is identified with 
$A'=\{(1,a_1),\ldots,(n,a_n)\}$.


\begin{definition}[Cyclic Permutation]
A \emph{cyclic permutation} on a (multi-)set $A$ with $n$ elements is a function
$\ri{C}\colon A\rightarrow A$
such that  
$\forall a\in A\colon\forall i\in\mathbb{N}^+\colon\ri{C}^i(a)= a ~\Leftrightarrow~ i\bmod n=0$.\footnote{If 
the condition holds for one element in $A$, it holds for all elements in $A$.}
We call $\ri{C}(a)$ the \emph{successor of $a$ (in $\ri{C}$)} and we write
$a\succ_\ri{C} \ri{C}(a)$.
\end{definition}

Every cyclic permutation is a bijective mapping.
Let $\ri{C}$ be a cyclic permutation on a (multi-)set $A=\{a_1,a_2,\ldots,a_n\}$.
One could imagine $\ri{C}$ as an arrangement of $A$ on a clock face, for an $n$-hour clock.
\emph{To go through $\ri{C}$ beginning at $a_i\in A$} means to consider the elements in 
$A$ as they appear in $\ri{C}$, beginning at $a_i$. Next we define subsets of $A$ that appear 
consecutively in $\ri{C}$.

\begin{definition}[Ring-Interval, Ring-Sequence
]\label{interval}
For a cyclic permutation $\ri{C}$ on a (multi-)set $A$ and $a,b\in A$, we recursively define 
the \emph{ring-interval from $a$ to $b$ (in $\ri{C}$)} as
$$[a,b]_\ri{C}\colonequals 
\begin{cases}
\{a\}&\text{if }a=b\\
\{a\}\cup[\ri{C}(a),b]_\ri{C}&\text{else.}\\
\end{cases}
$$
By $]a,b]_\ri{C}$ we indicate that the left endpoint is excluded, by  $[a,b[_\ri{C}$ that the right endpoint is
excluded, and by $]a,b[_\ri{C}$ that both endpoints  are excluded.
Therefore we have $]a,a],[a,a[,]a,a[=\emptyset$

The \emph{ring-sequence from $a$ to $b$ (in $\ri{C}$)}, denoted by 
$\cordered[a,b]_\ri{C}=(a,\ldots,b)$, is the order in which elements of $[a,b]_\ri{C}$ appear when going through $\ri{C}$ beginning at $a$.


\end{definition}


\begin{construction}\label{ringkombi}
Let $\ri{C}$ be a cyclic permutation on a (multi-)set $A$ and $a,b,c,d\in A$.
If $[a,b]_\ri{C}\cap[c,d]_\ri{C}=\emptyset$ and $b\succ_\ri{C}c$, we 
have $[a,b]_\ri{C}\cup[c,d]_\ri{C}=[a,d]_\ri{C}$.
\end{construction}

\subsection{Circular-Arc Graphs}\label{sub:CAG}

We assume that the reader is familiar with basic graph theoretical definitions.
The \emph{intersection graph} of a family of sets is the graph where the vertices are the sets, 
and the edges are the pairs of sets that intersect. Every graph is the intersection graph of 
some family of sets. A graph is an \emph{interval graph} if it is the intersection graph of a 
finite set of intervals (line segments) on a line and a \emph{unit interval graph} if these 
intervals have unit length.
A graph $G$ is a \emph{circular-arc graph} if it is the intersection graph of a finite set of 
arcs on a circle; the latter we call an \emph{arc model of $G$}.
Since an interval graph is a special case of a circular-arc graph, namely a circular-arc graph 
that can be represented with a set of arcs that do not cover the entire circle, we define the 
set of \emph{strict circular-arc graphs} as the set of circular-arc graphs that are not interval 
graphs. Every strict circular-arc graph is connected.
A \emph{unit circular-arc graph} is a circular-arc graph that has an arc model in which the 
arcs have unit length. For a survey on circular-arc graphs see \cite{Lin20095618} 
and for the definition of further special graph classes see 
\cite{BLS99}.
Fig. \ref{fig:animals} illustrates a strict circular-arc graph.

\begin{figure}[ht]
  \centering
  \subfloat[Arc model]{\label{fig:gull}\includegraphics[width=0.33\textwidth,trim=3cm 4cm 3cm 4cm]{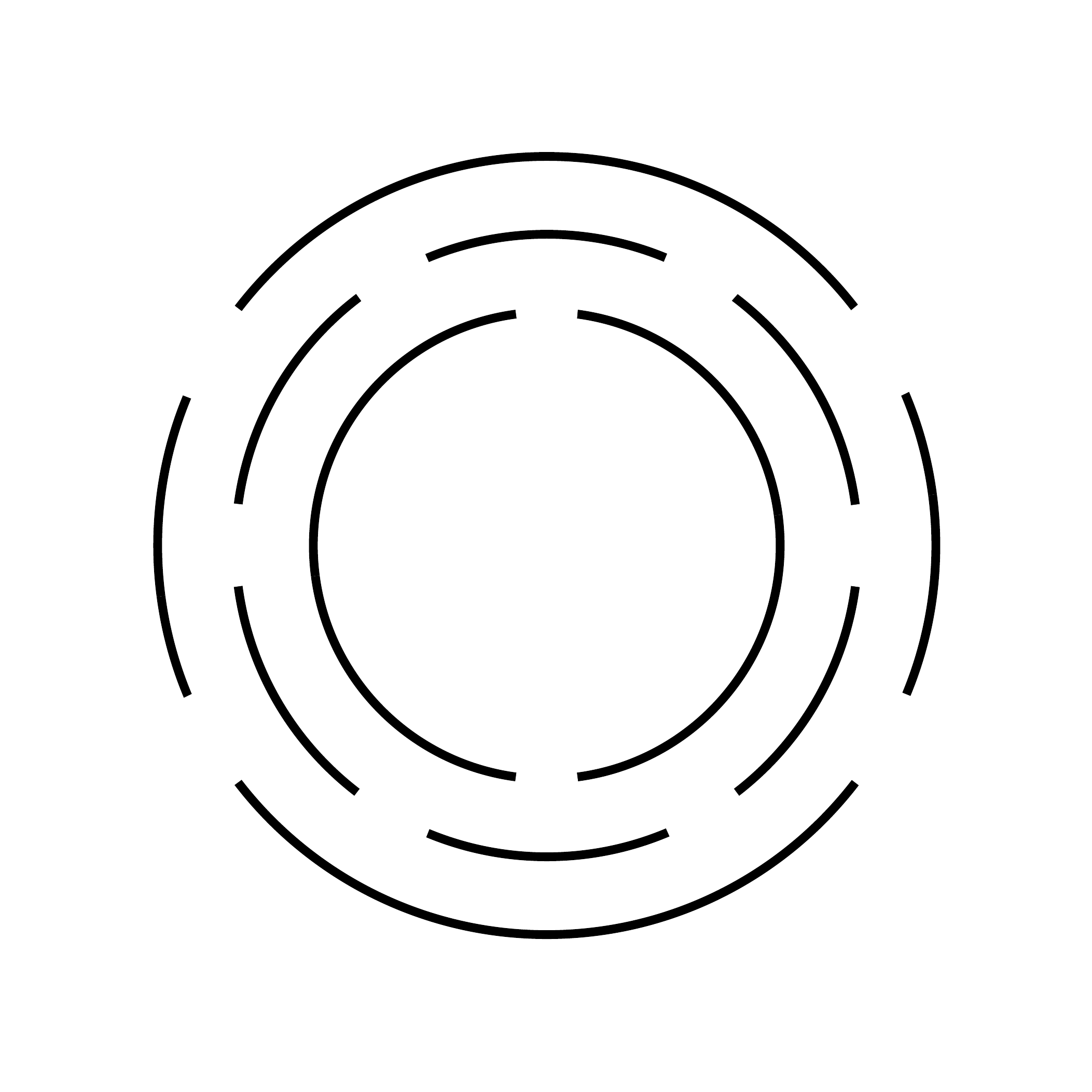}}                
  \subfloat[Correspondence]{\includegraphics[width=0.33\textwidth,trim=3cm 4cm 3cm 4cm]{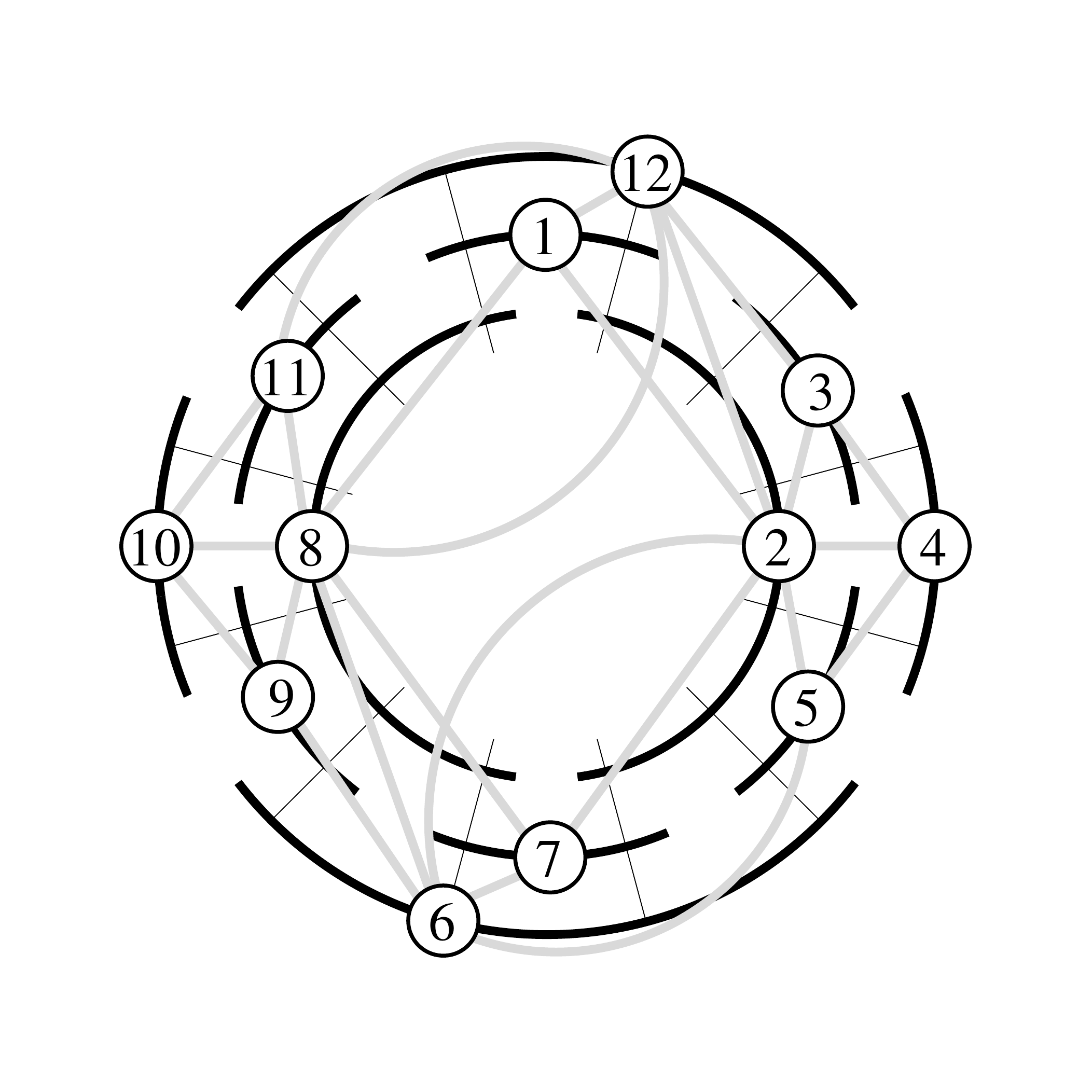}}
  \subfloat[Graph representation]{\label{fig:mouse}\includegraphics[width=0.33\textwidth,trim=3cm 4cm 3cm 4cm]{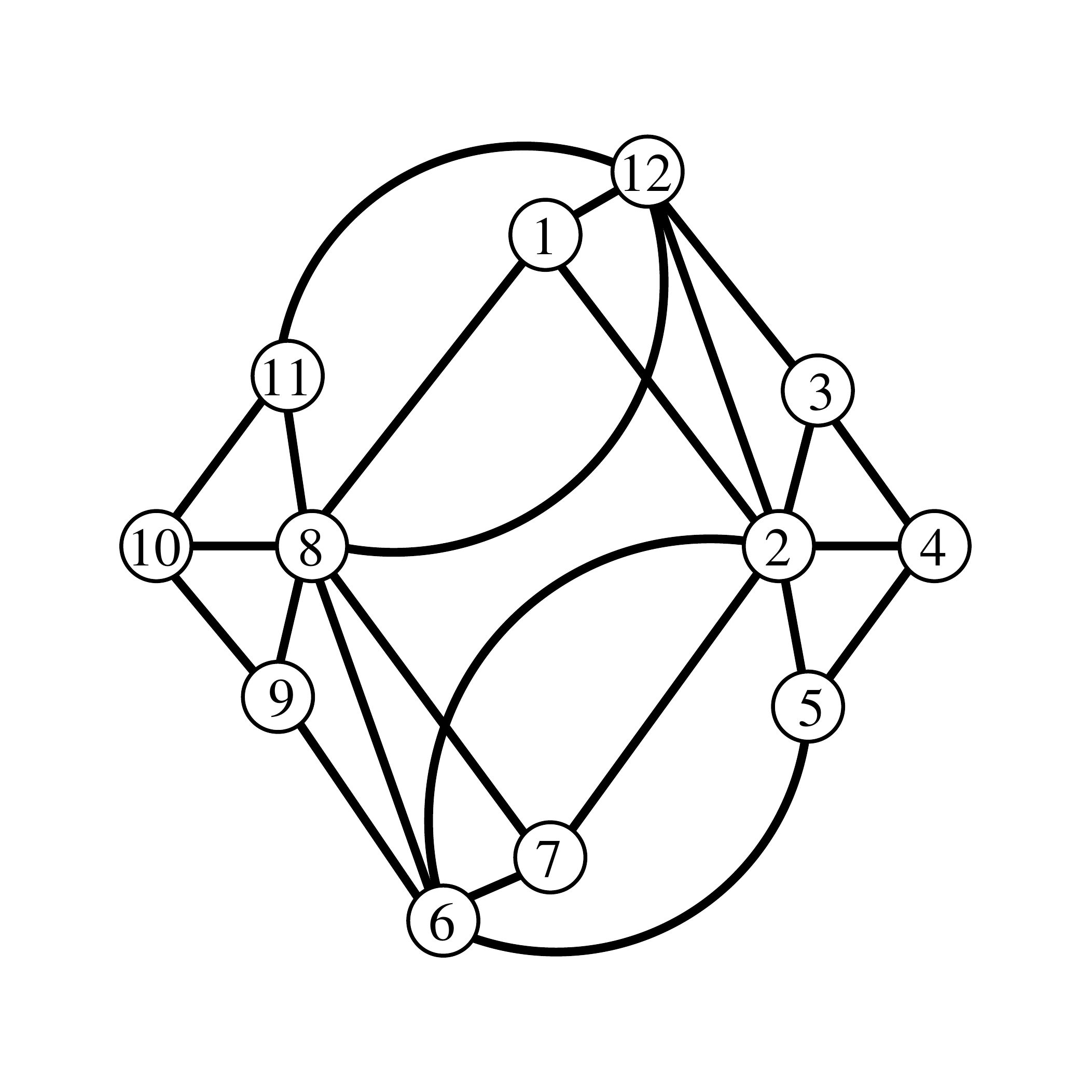}}
   \vspace{-5pt}\caption{A (strict) circular-arc graph}
  \label{fig:animals}
\end{figure}

A difference between interval graphs and circular-arc graphs, that is worth mentioning, 
is that the maximal cliques of interval graphs can be associated to points of 
the \q{interval model} and therefore an interval graph can have no more maximal cliques 
than vertices. In contrast, circular-arc graphs may contain maximal cliques that do not 
correspond to points of  some arc model \cite{Lin20095618}.
In fact, just as in arbitrary graphs, the number of maximal cliques in circular-arc 
graphs can grow exponentially in the size of the graph \cite{firstCAG}. For
the special case that there is an arc model where no three arcs cover the whole
circle the number of maximal cliques is bounded by the
number of vertices of the graph \cite{CFZ08}. Also maximal cliques may
occur several times within some arc model, which is addressed by multisets in 
this paper.

The following definition and the subsequent corollary formalize important 
graph-theoretic properties of circular-arc graphs.

\begin{definition}[Clique-Cycle, Left Clique $\lc{v}$, and Right Clique $\rc{v}$]\label{lcrc}
Let $G=(V,E)$ be a circular-arc graph, $M_G$ be an arc model of $G$, and $C$ be the 
circle of $M_G$.
To each point $p$ on $C$ corresponds a clique that contains the vertices 
whose corresponding arcs contain $p$.
Let multiset $\mc{Y}$ contain all cliques that correspond 
to points on $C$ and define $\mc{X}\subseteq\mc{Y}$ as the 
multiset of cliques we obtain, when we remove all cliques from $\mc{Y}$ that 
are not maximal with respect to inclusion among the cliques in $\mc{Y}$  and all 
but one cliques that are equal.
Pair $(\mc{X},\ri{C})$, where $\ri{C}$ is the cyclic permutation on $\mc{X}$ 
that is implied by ordering the elements in $\mc{X}$ as their corresponding points appear 
(clockwise) on $C$, is called a \emph{clique-cycle} for $G$. 
For every vertex $v\in V$, there
exist two cliques $\mc{a},\mc{b}\in\mc{X}$ such that a 
clique $\mc{x}\in\mc{X}$ contains $v$ if and 
only if $\mc{x}\in[\mc{a},\mc{b}]_\ri{C}$.
We call $\mc{a}$ the \emph{left clique of $v$}, denoted by $\lc{v}$, and $\mc{b}$ 
the \emph{right clique of $v$}, denoted by $\rc{v}$.
\end{definition}

When circular-arc graphs are discussed in this paper, the argumentation is based 
on one of their clique-cycles deduced from some arc model.

A frequently used notation in our paper is that of a {\em dominating vertex}, which
is a vertex adjacent to all other vertices of the graph.  The set of all
dominating vertices of graph $G$ is denoted by $\dom$.
Since a dominating vertex $d$ is adjacent to all the vertices, 
all the maximal cliques contain it, wherefore $d$'s left and right clique form a cyclic interval that covers the entire clique cycle.

\begin{corollary}\label{remarkDoms}
Let $(\mc{X},\ri{C})$ be a clique-cycle for a circular-arc graph $G=(V,E)$ 
and $d\in \dom$ a dominating vertex.
The left clique $\lc{d}$ and the right clique $\rc{d}$ of $d$ are not unique, 
but can be chosen arbitrarily with 
$\rc{d}\succ_\ri{C}\lc{d}$ being the only constraint.
\end{corollary}

Consider the arc model of a circular-arc graph and two intersecting arcs $v,w$.
The arcs $v$ and $w$ can intersect in the following two ways:
\begin{inparaenum}
\item
The intersection of $v$ and $w$ constitutes another arc.
\begin{inparaenum}
This is the case,
\item
if $v$ is included in $w$ or vice versa,
\item
if $v$ and $w$ are congruent, 
i.e., cover the exact same part of the cycle, or 
\item
if exactly one endpoint of $v$ lies in $w$ and vice versa.
\end{inparaenum}
\item
The intersection of $v$ and $w$ constitutes two arcs (the corresponding vertices will 
be defined as counter vertices in the subsequent definition).
This is the case, if $v$ and $w$ jointly cover the circle and meet at both ends.
\end{inparaenum}

\begin{definition}[Counter Vertex $\co{v}$]\label{reaches}
Let $(\mc{X},\ri{C})$ be a clique-cycle for a circular-arc graph $G=(V,E)$.
We call $w\in V$ a \emph{counter vertex} of $v\in V$, if $\lc{v},\rc{v}\in[\lc{w},\rc{w}]_\ri{C}$ and $\lc{w},\rc{w}\in[\lc{v},\rc{v}]_\ri{C}$ and $\lc{v}\neq{\lc{w}}$ (or, equivalently, $\rc{v}\neq{\rc{w}}$).
We call $(v,w)$ a \emph{pair of counter vertices} and denote the set of all counter 
vertices of $v$ by $\co{v}$.\footnote{Note that, $w\in\co{v}~\Leftrightarrow~v\in\co{w}$.}
\end{definition}

The next corollary follows directly from the arc model of a circular-arc graph.

\begin{corollary}\label{adjacentToVorCoV}
Let $G=(V,E)$ be a circular-arc graph and $v\in V$ be a vertex 
with $\co{v}\neq\emptyset$. Then every vertex in $V$ is adjacent 
to $v$ or to every vertex in $\co{v}$, in other words, $v$ and 
any of its counter vertices constitute a dominating set.
\end{corollary}

\begin{definition}[Reaching to the Left/Right]\label{defi:reaches}
Let $(\mc{X},\ri{C})$ be a clique-cycle for a circular-arc graph $G=(V,E)$, 
$v,w$ two adjacent vertices with $v\notin\co{w}$ and $\mc{x}\in \mc{X}$ a clique with $v,w\in \mc{x}$.
We say, \emph{$w$ reaches at least as 
far to the right as $v$} if $[\mc{x},\rc{v}]_\ri{C}\subseteq [\mc{x},\rc{w}]_\ri{C}$, 
and \emph{$w$ reaches further to the right than $v$} 
if $[\mc{x},\rc{v}]_\ri{C}\subset [\mc{x},\rc{w}]_\ri{C}$.
Analogously, we define \emph{$w$ reaches at least as 
far to the left as $v$} by $[\lc{v},\mc{x}]_\ri{C}\subseteq [\lc{w},\mc{x}]_\ri{C}$ 
and \emph{$w$ reaches further to the left than $v$} 
by $[\lc{v},\mc{x}]_\ri{C}\subset [\lc{w},\mc{x}]_\ri{C}$.
\end{definition}

Although the arcs corresponding to counter vertices 
cover very different areas of the circle, it is impossible to say 
which of the two arcs \q{reaches further to the left or right}, when the point of 
view is the middle of the circle.

\subsection{Interval Routing Schemes}\label{IRS}

We assume that the reader is familiar with basic concepts of interval routing and refer to \cite{Gavoille00asurvey} for an exhaustive introduction and survey.
If an IRS assigns at most $k$ intervals to each (directed) edge and only implies shortest paths, we denote it by $k$-IRS.
If this $k$-IRS is furthermore strict, which means that for every vertex $v\in V$ no interval assigned to the outgoing edges of $v$ contains $v$'s number, it is denoted by $k$-SIRS.
Definition \ref{fe} is needed for Definition \ref{defiIRS}, which is an alternative definition for shortest path strict interval routing schemes that is equivalent to definitions in literature but better suited for our purposes.
To apply Definition \ref{defiIRS} to an undirected graph (as for example a circular-arc graph), the graph is converted to its directed, symmetric version.
While directed edges are often called arcs, we choose the former term to avoid confusion with the arcs of a graph's arc model.

\begin{definition}[First Vertex $\fe{u}{w}$]\label{fe}
Let $G=(V,A)$ be a directed graph and $u,v,w\in V$.
Vertex $v\in V$ is a \emph{first vertex from $u$ to $w$}, if there exists a shortest directed path 
$(u,v,\ldots, w)$, i.e., there is a shortest directed path from $u$ to $w$, that first traverses $v$.
The set of first vertices from $u$ to $w$ is denoted by 
$$\fe{u}{w}=\{v\in V~|~v\text{ is a first vertex from }u\text{ to }w\}.$$
\end{definition}

\begin{definition}[Shortest Path Strict Interval Routing Scheme]\label{defiIRS}
~ Let $G=(V,E)$ be a directed graph.
A \emph{shortest path strict interval routing scheme for $G$} 
is a pair $R=(\ri{L},\ema{I})$, where $\ri{L}$ is a cyclic permutation on $V$, 
called the \emph{vertex order}, and $\ema{I}$, called the \emph{directed-edge-labeling}, maps 
every directed edge to a set of ring-intervals in $\ri{L}$ such that for every vertex $v\in V$,
\begin{enumerate}

\item 
$\ema{I}$ maps the outgoing directed edges of $v$ to a set of ring-intervals 
in $\ri{L}$, such that the intervals
assigned to different directed edges never intersect,

\item
for every vertex $u\neq v$  one of these these ring-intervals 
contains $u$ (vertex $v$ must not appear in one of these intervals), and

\item
if $u$ is contained in a ring-interval in $\ema{I}((v,w))$, then $w$ is a first vertex 
from $v$ to $u$.
\end{enumerate}

Let $\ri{F}$ be a vertex order.
We say that  \emph{$v\in V$ given $\ri{F}$ suffices a shortest path $k$-SIRS}, if  
there exists a directed edge-labeling for $G$ that maps every outgoing directed edge of $v$ to at most 
$k$ ring-intervals in $\ri{F}$, such that $v$ satisfies the three constraints above.
\end{definition}


\begin{corollary}\label{coro:neu}
A directed graph $G=(V,A)$ supports a 
shortest path $k$-SIRS, if a vertex order $\ri{F}$ exists, such that every 
vertex $v\in V$  given $\ri{F}$ suffices a shortest path $k$-SIRS.
\end{corollary}

%


\section{Main result\label{main}}

The compactness of the class of circular-arc graphs has not been determined until now, 
which is somehow surprising, since the class of circular-arc graphs is closely related 
to the class of interval graphs and unit circular-arc graphs and both classes are known 
to have strict compactness $1$ \cite{NS96,FG98}. 
However, for the compactness of the class of circular-arc graphs only a lower bound of 2 was known.
In particular, circular-arc graphs exist that do not allow for optimal 1 interval routing schemes, as, for example, a wheel graph, i.e., a cycle together with a dominating vertex, with six outer vertices \cite{FG98}.
Also the circular-arc graph shown in Fig. \ref{fig:animals} does not allow optimal 1 interval routing schemes.
In this section we prove the main result of this paper, which is given by the following theorem and shows that the lower bound of 2 for the compactness of circular-arc graphs is indeed sharp.

\begin{theorem}[Main Theorem]\label{mainTheo}
The class of circular-arc graphs has strict compactness 2.
\end{theorem}

Since every non-strict circular-arc graph $G$ is an interval graph and therefore has strict compactness 1 \cite{NS96},  we only need to consider strict circular-arc graphs in the proof.
For the rest of this section let $(\mc{X},\ri{C})$ be a clique-cycle for an 
arbitrary strict circular-arc graph $G_u=(V,E)$ and $G=(V,A)$ be the directed symmetric 
version of $G_u$. 

\paragraph*{Outline of the proof}
To show that there exists a shortest path $2$-SIRS for $G$ and thus 
Theorem \ref{mainTheo} holds, we proceed as follows. 
In Section \ref{constrCO} we construct 
a cyclic permutation $\ri{L}$ on $V$.
In Section \ref{furtherDefs} we choose an arbitrary vertex $v\in V$ and partition the 
vertices in $V-\{v\}$ in three disjoint ring-intervals $A_v,B_v,C_v$ in $\ri{L}$ in 
dependency on $v$. 
In order to show that $v$ given $\ri{L}$ suffices a shortest path 2-SIRS, 
we show how to define an directed edge-labeling $\ema{I}_v$ for the outgoing directed edges of $v$ that 
satisfies the corresponding  constraints given in Definition \ref{defiIRS}. We
consider each of the three ring-intervals $A_v,B_v,C_v$ in 
Section  \ref{lvbisv} and \ref{mvbislv} and 
show how the vertices in the respective ring-interval can be sufficiently  mapped to 
by $\ema{I}_v$. 
Since we have chosen $v$ arbitrarily, it follows that every vertex in $V$ 
suffices a shortest path $2$-SIRS given $\ri{L}$ and therefore, by 
Corollary \ref{coro:neu}, $G$ supports a shortest path $2$-SIRS.
Since $G_u$ is an arbitrary strict circular-arc graph, this proves 
Theorem \ref{mainTheo}.

The following notations are straight forward but may be formalized for convenience.
Let $(v,w)\in E$ be a directed edge with $\ema{I}_v((v,w))= S$, where $S$ is a set of ring-intervals 
in $\ri{L}$, and $R$ be a ring-interval in $\ri{L}$. \emph{Assigning $R$ to directed edge $(v,w)$} 
means to define $\ema{I}_v\colon(v,w)\mapsto S\cup \{R\}$.
Of course, when we start constructing $\ema{I}_v$, every directed edge is mapped to the empty set.
Let $[a,b]_\ri{L}$ and $[c,d]_\ri{L}$ be two ring-intervals that can be joined to one 
ring-interval $[a,d]_\ri{L}$ by Construction \ref{ringkombi}. If $[a,b]_\ri{L}$ and 
$[c,d]_\ri{L}$ are assigned to the same directed edge $(v,w)$, we can redefine
$$\ema{I}_v\colon(v,w)\mapsto\{[a,b]_\ri{L},[c,d]_\ri{L}\} \text{ by }
\ema{I}_v\colon(v,w)\mapsto\{[a,d]_\ri{L}\}$$
and therefore save one ring-interval on directed edge $(v,w)$.
This redefinition is called \emph{compressing the (two) ring-intervals on directed edge $(v,w)$ 
(to one ring-interval)}. Let $V'\subseteq V$ be a set of vertices. 
\emph{To distribute $V'$ (over the outgoing directed edges of $v$)} means to partition 
$V'$ into ring-intervals and assign these to outgoing directed edges of $v$ (sufficiently).
For convenience,
we sometimes refer to ring-intervals as \emph{intervals}.

\subsection{Definition of the Vertex Order}\label{constrCO}


In this section we consider a circular-arc graph $G=(V,E)$ 
together with a clique-cycle $(\mc{X},\ri{C})$ and show
how to construct a cyclic permutation $\ri{L}$ on $V$ 
that serves as the given vertex order. The ordering is obtained by
sorting the arcs using their left cliques as primary sort key
and right cliques as secondary sort key, making the first
arc the successor of the last. A pseudocode for this 
purpose is presented in Listing \ref{fig:animals} and its idea is explained 
next.

By Corollary \ref{remarkDoms} the left and right clique of a dominating vertex are not unique.
Since it simplifies the proof, if all dominating vertices have the same left and thus also 
the same right clique, these cliques are unified in Line 1.
For a fixed $\mc{z}$ in Line 1, the generated cyclic permutation is fully deterministic, 
except for the ordering of true twins.\footnote{Two vertices in a graph are called 
\emph{true twins} if they are adjacent to the same set of vertices and to each other.}
The dummy vertex introduced in Line 4 is needed to close the cyclic order once all vertices 
are integrated in $\ri{L}$.
The loop in Line 6 runs once through all cliques in $\mc{X}$ (beginning at $\mc{f}$, that 
is arbitrarily chosen in Line 2) in the order defined by $\ri{C}$ (Line 13).
For every visited clique $\mc{n}$ the loop in Line 8 integrates every vertex, whose left 
clique is $\mc{n}$, in $\ri{L}$.
By Line 9, vertices with the same left clique are ordered with respect to 
their right clique. An example for the defined vertex ordering can be found in 
Fig. \ref{fig:animals}.

\begin{lstlisting}[caption={Definition of a cyclic order $\ri{L}$ on $V$ that serves as the given vertex order.},captionpos=b,label=pseudoLabeling]
Fix $\mc{z}\in \mc{X}$ arbitrarily and choose $\lc{d}\colonequals \ri{C}(\mc{z})$ and $\rc{d}\colonequals \mc{z}$ for all $d\in \dom$
Choose $\mc{f}\in \mc{X}$ arbitrarily
$\mc{n}\colonequals \mc{f}$
$V\colonequals V\cup \{t\}$    /* add dummy vertex $t$ */
$p\colonequals t$
do
      $A\colonequals \{v\in V~|~\lc{v}=\mc{n}\}$
      while($A\neq\emptyset$)
	    Choose $a\in A$ such that every vertex in $A$ reaches at least as far to the right as $a$.
	    $\ri{L}(p)\colonequals a$
	    $p\colonequals a$
	    $A\colonequals A-\{a\}$
      $\mc{n}\colonequals \ri{C}(\mc{n})$
while($\mc{n}\neq\mc{f}$) 
$\ri{L}(p)\colonequals\ri{L}(t)$
$V=V-\{t\}$     /* remove dummy vertex $t$ */
return $\ri{L}$
\end{lstlisting}


The following observations are crucial for the rest of this paper.


\begin{remark}
\begin{enumerate}
\item\label{item:leftconsecutiv}
The vertices  are ordered primarily by their left clique, that is to say, vertices having 
the same left clique appear consecutively in $ \ri{L}$.

\item 
Vertices with the same left clique are ordered in ascending 
order by their right clique.

\item\label{item:onemoreleft}
For two vertices $v,w\in V$ with $v\succ_\ri{L} w$, we know that $v$ and $w$ are adjacent and either
\begin{enumerate}

\item
$\lc{v}=\lc{w}$ and $w$ reaches at least as far to the right as $v$, or

\item
$\lc{v}\succ_\ri{C}\lc{w}$, in other words, $v$ reaches one clique further to the 
left than $w$.%
\footnote{Because the relation \q{reaching to the right} is not defined, if $v$ and $w$ are counter vertices, we intuitively extend Definition \ref{defi:reaches} to this case by fixing $\mc{x} = \lc{w}$.}

\end{enumerate}
\end{enumerate}
\end{remark}

The next definition is based on Assertion \ref{item:leftconsecutiv} of the preceding 
remark, that is to say, on the fact that, since vertices having the same left clique 
appear consecutively in $\ri{L}$, for a given clique $\mc{c}\in\mc{X}$, two 
vertices $v,w\in \mc{c}$ exist, such that the vertices in  ring-interval $[v,w]_\ri{L}$ are 
exactly those with left clique $\mc{c}$.

\begin{definition}[Head Vertex $\first{\mc{c}}$ and Tail Vertex $\last{\mc{c}}$]
Let $(\mc{X},\ri{C})$ be a clique-cycle for some circular-arc graph $G=(V,E)$, $\ri{L}$ 
given by Listing \ref{fig:animals}, $\mc{c}\in\mc{X}$, and $v,w\in \mc{c}$ the unique 
vertices such that $[v,w]_\ri{L}=\{x\in V~|~\lc{x} = \mc{c}\}$ ($v$ and $w$ are not necessarily distinct).
We call $v$ the \emph{head vertex of $\mc{c}$} and $w$ the \emph{tail vertex of $\mc{c}$} 
and denote them by $\first{\mc{c}}$ and $\last{\mc{c}}$, respectively.
\end{definition}


\begin{corollary}\label{coro:tail vertexheal}
Let $(\mc{X},\ri{C})$ be a clique-cycle for a circular-arc graph $G=(V,E)$, $\ri{L}$ given 
by Listing \ref{pseudoLabeling}, and $\mc{c}\in\mc{X}$.
The successor (in $\ri{L}$) of the tail vertex of $\mc{c}$ is the head vertex of the 
successor (in $\ri{C}$) of $\mc{c}$, i.e.,
$\ri{L}(\last{\mc{c}})=\first{\ri{C}(\mc{c})}$.
\end{corollary}

\subsection{Partition of the Vertex Order}\label{furtherDefs}

The cyclic permutation $\ri{L}$ on $V$ allows to state further definitions.
We now fix $v\in V$ arbitrarily for the rest of the proof and show that $v$ given $\ri{L}$ 
suffices a shortest path $2$-SIRS.
If $v$ is a dominating vertex, we can define $\ema{I}_v\colon(v,w)\mapsto[w,w]_\ri{L}$ 
for every vertex $w\in V-\{v\}$ and have thereby shown that $v$ given $\ri{L}$ suffices a shortest path $1$-SIRS.
Therefore, we now assume that $v$ is not a dominating vertex.
Fig. \ref{fig:tiger3} illustrates the subsequent Definition.

\begin{definition}[Left Vertex $\lvv$]\label{de_l_v}
Let $(\mc{X},\ri{C})$ be a clique-cycle for some circular-arc graph $G=(V,E)$, $\ri{L}$ 
given by Listing \ref{fig:animals}, and $v\in V$ not a dominating vertex.
\begin{inparaenum}[(i)]
Let $L\subset V$ be the union of the following two vertex sets.
\item The set of
vertices that are adjacent to $v$ and reach further to the left than $v$ 
but are neither counter vertices of $v$ nor dominating vertices.
\item
The set of vertices in $[\first{\lc{v}},v[_\ri{L}$.
\footnote{We have $\dom\cap[\first{\lc{v}},v[_\ri{L} = \emptyset$ and $\co{v} \cap [\first{\lc{v}},v[_\ri{L} = \emptyset$.}
\end{inparaenum}
If $L\neq\emptyset$, we define the \emph{left vertex of $v$}, denoted by $\lvv$,  
as the unique vertex in $L$ such that
\begin{equation}\label{eq_l_v}
\forall w \in L\colon [w,v[_\ri{L}\subseteq[\lvv,v[_\ri{L} \text{ (or, equivalently, } L\subseteq[\lvv,v[_\ri{L}\text{)}.
\end{equation}
\end{definition}

\begin{figure}
\centering
\subfloat[The green vertices are in set $L$ of Definition \ref{de_l_v}. 
The left vertex of $v$, which is also in $L$, is colored in light green 
and the middle vertex of $v$ in red.]{\label{fig:tiger3}\includegraphics[trim=0cm 2cm 0cm 0cm, width=0.45\textwidth,
  ]{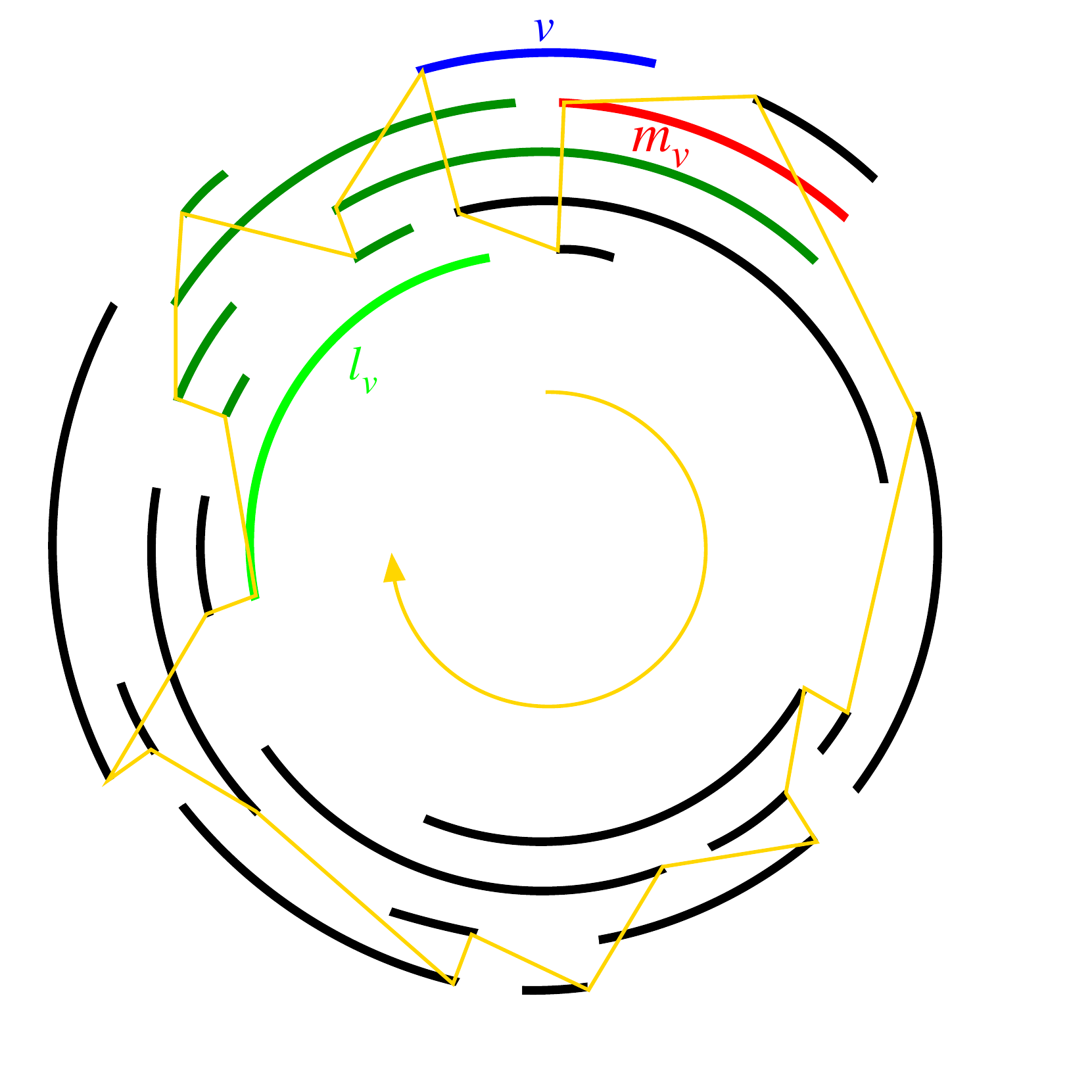}}\quad
  \subfloat[The green vertices are to the left of $v$, the red vertices are to the right of $v$, and the black vertices are face-to-face with $v$. 
]{\label{fig:gull3}\includegraphics[trim=0cm 2cm 0cm 0cm, width=0.45\textwidth,
  ]{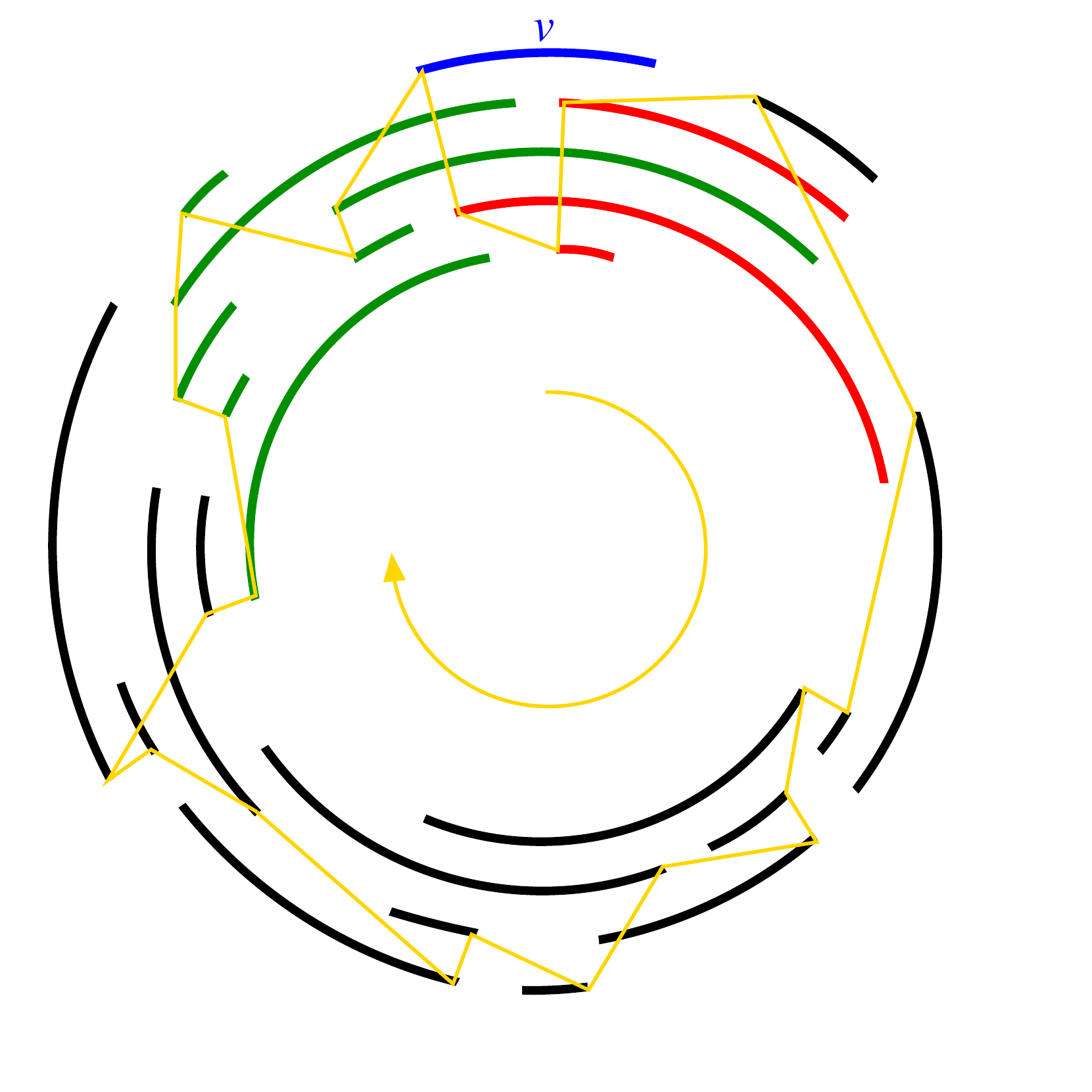}}                
  \vspace{-5pt}\caption[An example for the vertex partition defined in Section \ref{furtherDefs}]{Two examples for the Definitions made in this section
The fixed vertex $v$ is colored blue. 
The yellow line connects the vertices as they appear in vertex order $\ri{L}$ as generated by Listing \ref{fig:animals}.
}
\label{fig:anotherExamplefurtherDefs}
\end{figure}

\begin{theorem}\label{no_vertex_further_than_lv}
Let $(\mc{X},\ri{C})$ be a clique-cycle for a circular-arc graph $G=(V,E)$, $\ri{L}$ 
given by Listing \ref{fig:animals}, $v\in V$ not a dominating vertex, $\lvv$ the left 
vertex of $v$, and $u\neq\lvv$ be a vertex that is adjacent to $v$ and reaches further to the left than $v$ but is neither 
a dominating vertex nor a counter vertex of $v$. 
Then $\lvv$ reaches at least as far to the left as $u$.
\end{theorem}

\begin{proof}
Assume $u$ reaches further to the left than $\lvv$.
Then by Definition \ref{defi:reaches} we know that 
$[\lc{\lvv},\lc{v}]_\ri{C}\subset[\lc{u},\lc{v}]_\ri{C}$, which implies 
that $$\cordered[\lc{u},\lc{v}]_\ri{C}=(\lc{u},\ldots,\lc{\lvv},\ldots,\lc{v}).$$
Since $\ri{L}$ is primary ordered by the left cliques of the vertices, it follows 
$$u\succ_\ri{L}\ldots\succ_\ri{L} \lvv\succ_\ri{L}\ldots\succ_\ri{L}v$$ and therefore 
$[\lvv,v]_\ri{L}\subset[u,v]_\ri{L}$.
This contradicts (\ref{eq_l_v}) in Definition \ref{de_l_v}, since $u\in L$ reaches 
further to the left than $v$ and is neither a 
dominating vertex nor a counter vertex of $v$.
\end{proof}

In Definition \ref{leftrightface} we partition the vertices in $V-\{v\}$ into three 
ring-intervals in $\ri{L}$.
The following argumentation gives the idea behind this partitioning and also proves 
Theorem \ref{Lvrl_v}. Consider we go through $\ri{L}$ beginning at $v$.
Since the vertices in $\ri{L}$ are primary ordered by their left cliques, the left 
clique of the vertices first considered is $\lc{v}$ (or $\ri{C}(\lc{v})$, if $v$ is the 
tail vertex of clique $\lc{v}$).
The next vertices traversed have left clique $\ri{C}(\lc{v})$, followed by vertices with 
left clique $\ri{C}^2(\lc{v})$, and so on.
We eventually reach the set of vertices with left clique $\rc{v}$.
The last vertex we come across in this set is $\last{\rc{v}}$ and fixed in Definition \ref{de_m_v}.
The vertices we came across so far (excluding $v$, including $\last{\rc{v}}$) will be 
defined as the \q{vertices to the right of $v$} in Definition \ref{leftrightface} and are 
adjacent to $v$, as $v$ was contained in their left clique.

\begin{definition}[Middle Vertex $\mvv$]\label{de_m_v}
Let $(\mc{X},\ri{C})$ be a clique-cycle for a cir\-cular-arc graph $G=(V,E)$, $\ri{L}$ 
given by Listing \ref{fig:animals}, and $v\in V$ not a dominating vertex.
We define the \emph{middle vertex of $v$}, denoted by $\mvv=\last{\rc{v}}$,  
as the tail vertex of the right clique of $v$.
\end{definition}

When we continue visiting the vertices after $\mvv$ in the same manner, the next vertex $u$ we come across that is not a dominating vertex but adjacent to $v$ has to be $\lvv$, as we could find a contradiction in the same manner as 
in the proof of Theorem \ref{no_vertex_further_than_lv}, if we had $u\neq\lvv$.
These vertices after $\mvv$ and before $\lvv$ are defined as the \q{vertices face-to-face with $v$} 
in Definition \ref{leftrightface}. When we continue to go through $\ri{L}$, we eventually 
reach $v$ again. These vertices from $\lvv$ to $v$ (excluding $v$) are defined 
as the \q{vertices to the left of $v$} in Definition \ref{leftrightface}, which is illustrated in Fig. \ref{fig:gull3}.

\begin{definition}[$A_v$, $B_v$, $C_v$]\label{leftrightface}
Let $(\mc{X},\ri{C})$ be a clique-cycle for a circular-arc graph $G=(V,E)$, $\ri{L}$ 
given by Listing \ref{fig:animals}, $v\in V$ not a dominating vertex, $\lvv$ the left 
vertex of $v$, if $l_v$ exists, and $\mvv$ the middle vertex of $v$. We define
\begin{enumerate}
\item
the \emph{vertices to the right of $v$} 
as $A_v=]v,\mvv]_\ri{L}$, and 


\item
the \emph{vertices face-to-face with $v$}
\begin{enumerate}


\item
as $B_v=]\mvv,v[_\ri{L}$, if $\lvv$ does not exist, or else

\item
as $B_v=]\mvv,\lvv[_\ri{L}$, and
\end{enumerate}
\item 
the \emph{vertices to the left of $v$}%
\begin{enumerate}

\item 
as $C_v=[\lvv,v[_\ri{L}$, if $\lvv$ exists, or else

\item
as $C_v=\emptyset$.\footnote{
We have $C_v=L$ from Definition \ref{de_l_v}.
}
\end{enumerate}
\end{enumerate}
\end{definition}

The next theorem follows from the argumentation between 
Theorem \ref{no_vertex_further_than_lv} and Definition \ref{leftrightface}.

\begin{theorem}\label{Lvrl_v}
Let $(\mc{X},\ri{C})$ be a clique-cycle for a circular-arc graph $G=(V,E)$, $\ri{L}$ 
given by Listing \ref{fig:animals}, and $v\in V$ but not a dominating vertex.
\begin{enumerate}

\item \label{Lvrl_v1}
Every vertex to the right of $v$ is adjacent to $v$.

\item
Every vertex face-to-face with $v$ that is not a dominating vertex is not 
adjacent to $v$.
\end{enumerate}
\end{theorem}

\subsection{Definition of the Directed Edge-labeling}\label{2-app}

Next we show that $v$ given $\ri{L}$ suffices a shortest 
path $2$-SIRS by constructing a mapping $\ema{I}_v$ from the set of outgoing directed edges of $v$ 
to at most two ring-intervals in $\ri{L}$, according to Definition \ref{defiIRS}.
We investigate the vertices to the left of $v$ and the vertices 
to the right of $v$ in Section \ref{lvbisv} and the vertices face-to-face with $v$ 
in Section \ref{mvbislv}.
For every vertex $u$ in the respective interval, we determine a first vertex $w$ from 
$v$ to $u$ such that $u$ and every vertex that is hitherto assigned to directed edge $(v,w)$ can 
be embraced by at most two ring-intervals in $\ri{L}$.
This is a simple task for the vertices in $A_v$ and $C_v$ and 
even for the vertices in $B_v$ the approach is straight forward, if there is a 
dominating vertex or a counter vertex of $v$.
In fact, it only gets tricky, if there are no dominating vertices and no 
counter vertices.

\subsubsection{Vertices to the Left and Right}\label{lvbisv}

By the first assertion of Theorem \ref{Lvrl_v} every vertex to the right of $v$ is 
adjacent to $v$, wherefore these vertices can be distributed by assigning $[w,w]_\ri{L}$ 
to directed edge $(v,w)$, for every vertex $w\in A_v$. 
Assume that we have $C_v\neq\emptyset$.
In general, not every vertex in $C_v$ is adjacent to $v$. 
Although $\lvv$ is adjacent to all vertices in $C_v$, it is generally not possible to assign 
$C_v$ to directed edge $(v,\lvv)$, since $C_v$ may contain vertices that are adjacent to $v$ (other than $\lvv$), and therefore must be assigned 
to \q{their own directed edge}.
Let $v_i\in C_v$ be a vertex that is adjacent to $v$ and assume we start at $v_i$ to go through 
$\ri{L}$ until we come across the next vertex that is adjacent to $v$ or $v$ itself.
Denote this vertex $v_j$.
We can show that $v_i$ is adjacent to every in $]v_i,v_j]_\ri{L}$, wherefore we 
can assign $[v_i,v_j[_\ri{L}$ to directed edge $(v,v_i)$.
Therefore the vertices to the left of $v$ can be distributed by assigning every vertex 
$w\in C_v$ to directed edge $(v,w)$, if $w$ is adjacent to $v$, and else to directed edge $(v,w')$, where $w'$ 
is the first vertex that precedes $w$ in $\ri{L}$ and is adjacent to  $v$.
Theorem \ref{theototheleft} proves the outlined idea formally.
When distributing the vertices in sets $A_v$ and $C_v$ as just outlined, every 
outgoing directed edge of $v$ (except for edges incident to dominating vertices face-to-face with $v$, which are not yet labeled) gets one ring-interval assigned.

\begin{theorem}\label{theototheleft}
Let $(\mc{X},\ri{C})$ be a clique-cycle for a circular-arc graph $G=(V,E)$, $\ri{L}$ 
given by Listing \ref{fig:animals}, $v\in V$ but not a dominating vertex, and $\lvv$ 
the left vertex of $v$ (then the vertices in $[\lvv,v[_\ri{L}$ are the vertices to 
the left of $v$).
Let $\cordered[\lvv,v]_\ri{L}=(\lvv=v_1,v_2,\ldots,v_n=v)$ be the sequence of vertices that we obtain, when we order the vertices in $[\lvv,v[_{\ri{L}}$ that are adjacent to $v$ as they appear in ring-sequence $\cordered[\lvv,v[_\ri{L}$ and append $v$.
For $1\leq i<n$ and $x\in[v_i,v_{i+1}[_\ri{L}$, vertex $v_i$ is a first vertex 
from $v$ to $x$.
\end{theorem}

\begin{proof}
For some $i$, $1\leq i<n$, let $x\in[v_i,v_{i+1}[_\ri{L}$.
Clearly the theorem holds if $x=v_i$.
If $x\neq v_i$, $x$ does not appear in $(v_1,v_2,\ldots,v_n)$ and therefore is not 
adjacent to $v$. Since the vertices in $\ri{L}$ are primary ordered by their left clique 
and $v_i$ appears before $x$ in $\cordered[\lvv,v]_\ri{L}$, it follows that the left 
clique of $x$ is in $[\lc{v_i},\lc{v}]_\ri{C}$. Since $v_i$ is adjacent to $v$, $v_i$ 
is contained in every clique in $[\lc{v_i},\lc{v}]_\ri{C}$.
This implies that $v_i$ is contained in the left clique of $x$ and thus $v_i$ and $x$ 
are adjacent. Therefore $(v,v_i,x)$ is a shortest path between $v$ and $x$ and $v_i$ 
is a first vertex from $v$ to $x$.
\end{proof}

\subsubsection{Vertices Face-to-Face}\label{mvbislv}

This section discusses the distribution of vertices face-to-face with $v$.
It is probably the common case, that neither counter nor dominating vertices exist,
since only then path lengths are unbounded.
Thus, this case is discussed in the next paragraph and the 
subsequent paragraph discusses the remaining cases.

\paragraph{Neither Dominating nor Counter Vertices exist}\label{beideLeer}
In this case, $\lvv$ always reaches further to the left than $v$ and there is at least 
one vertex adjacent to $v$ that reaches further to the right.
Let $w\in V$ be a vertex that is not adjacent to $v$, $L$ be the set of vertices adjacent to $v$ 
that reach farthest to the left and $R$ be the set of vertices adjacent to $v$ that reach 
farthest to the right ($L$ and $R$ might intersect).
As evident from the arc model of $G$, every vertex in $L$ or every vertex in $R$ 
is a first vertex from $v$ to $w$ (cf. Corollary \ref{b}).
Since we have $\lvv\in L$, we now fix a vertex adjacent to $v$ that reaches farthest to the right.

\begin{definition}[Right Vertex $\rvv$]\label{de_r_v}
Let $(\mc{X},\ri{C})$ be a clique-cycle for a strict cir\-cu\-lar-arc graph $G=(V,E)$ without 
dominating vertices and without counter vertices, $\ri{L}$ given by Listing \ref{fig:animals}, 
$v\in V$, $\lvv$ the left vertex of $v$, $\mvv$ the middle vertex of $v$, and $R\subset V$ be the 
set of vertices that are adjacent to $v$ and reach farthest to the right.
We define the \emph{right vertex of $v$}, denoted by $\rvv$,  as $\lvv$, if $\lvv\in R$, or 
as $\mvv$, if $\mvv\in R$, or else as an arbitrary vertex in $R$.
\end{definition}
In order to show that $v$ suffices a shortest path 2-SIRS given $\ri{L}$, we could choose $\rvv$ arbitrarily in $R$ even if $\lvv\in R$ or $\mvv\in R$.
However, we will show that $v$ given $\ri{L}$ suffices a 
shortest path $1$-SIRS, if $\lvv\in R$ or $\mvv\in R$.
The next definition redefines the notation of the left and right vertex as a 
function, in order to allow 
recursive usage to easily determine a first vertex from $v$ to every vertex.
The subsequent corollary is clear when illustrated.

\begin{definition}[$\lvi{v}{}$,  $\rvi{v}{}$ and $\lvd{u}$, $\rvd{u}$]\label{lvi}
Let $(\mc{X},\ri{C})$ be a clique-cycle for a strict circular-arc graph $G=(V,E)$ without 
dominating vertices and without counter vertices, cyclic permutation 
$\ri{L}$ given by Listing \ref{fig:animals}, 
and $v\in V$. We define $\lvi{v}{}$ as the left vertex of $v$ and $\rvi{v}{}$ as the right vertex 
of $v$.\footnote{It follows that $\lvi{v}{2}$ is the left vertex of the left vertex of $v$ and $\rvi{v}{2}$ 
is the right vertex of the right vertex of $v$, etc.}
For $u\in V$, we define $\lvd{u}$ as the smallest $i$ such that $u$ is adjacent to $\lvi{v}{i}$ and  
$\rvd{u}$ as the smallest $i$ such that $u$ is adjacent to $\rvi{v}{i}$.
\end{definition}

\begin{corollary}\label{b}
Let $(\mc{X},\ri{C})$ be a clique-cycle for a strict circular-arc 
graph $G=(V,E)$ without dominating vertices and without counter 
vertices, cyclic permutation 
$\ri{L}$ given by Listing \ref{fig:animals}, 
and $v\in V$. The following three assertions hold for every 
vertex $u\in V$ that is not adjacent to $v$.
\begin{enumerate}
\item\label{jans}
If
$\lvd{u}<\rvd{u}$, then $\lvv$ is a first vertex from $v$ to $u$.
\item\label{jans2}
If
$\lvd{u}>\rvd{u}$, then $\rvv$ is a first vertex from $v$ to $u$.
\item\label{jans3}
If
$\lvd{u}=\rvd{u}$, then $\lvv$ and $\rvv$ are first vertices from $v$ to $u$.
\end{enumerate}
\end{corollary}

Since no vertex in $B_v$ is adjacent to $v$, the corollary implies that $\lvv$ or $\rvv$ is a 
first vertex from $v$ to every vertex in $B_v$.
Since sequence 
$\rvi{v}{1},\rvi{v}{2},\ldots$ \q{runs} through the clique-cycle as 
implied by $\ri{C}$ and sequence 
$\lvi{v}{1},\lvi{v}{2},\ldots$ \q{runs} through the clique-cycle as 
implied by $\ri{C}^{-1}$, that is to say, in the other direction, 
there exists an $i\in\mathbb{N}$ such that $\lvi{v}{i}$ and $\rvi{v}{i}$ \q{meet}, 
or more formally, are adjacent or equal. 
This number is fixed in the following definition.

\begin{definition}[Apex Number]\label{apex}
Let $(\mc{X},\ri{C})$ be a clique-cycle for a strict cir\-cular-arc graph $G=(V,E)$ 
without dominating vertices and without counter vertices, 
$\ri{L}$ given by Listing \ref{fig:animals}, and $v\in V$.
If we have $\lc{\lvi{v}{1}}=\rc{\rvi{v}{1}}$ or 
$[\lc{v},\rc{v}]_\ri{C}\subset[\rc{\rvi{v}{1}},\lc{\lvi{v}{1}}]_\ri{C}$ 
we set $i=1$, 
otherwise we define $i$ as the smallest number greater than 1 
such that $\lvi{v}{i}$ and $\rvi{v}{i}$ are adjacent or equal.
We call integer $i$ the \emph{apex number of $v$}.
\end{definition}

The separate treatment for apex number $i=1$ in Definition \ref{apex} ensures that the 
arcs of $\lvi{v}{i}$ and $\rvi{v}{i}$ intersect/meet adverse to $v$ on the 
circle and therefore the arcs of $\lvi{v}{1},\ldots,\lvi{v}{i}$, $\rvi{v}{1},\ldots,\rvi{v}{i}$,
and $v$ to cover the entire circle.
In particular, if $\lvv$ and $\rvv$ intersect in $v$'s arc they cannot 
cover the entire cycle, since they are not counter vertices, but the apex 
number would be defined as 1, if only the second part of the case distinction were in place.

By the definition, for apex number $i>2$ vertices $\rvi{v}{i}$ and $\lvi{v}{i-1}$ are 
adjacent if and only if $\lvi{v}{i}$ and $\rvi{v}{i-1}$ are adjacent.
Furthermore, for $2\leq j,k < i$ vertices $\rvi{v}{j}$ and $\lvi{v}{k}$ are not adjacent.

The following Theorem \ref{separator} shows that there 
exists a vertex $\svv\in B_v$ with $\lvi{v}{i},\rvi{v}{i}\in\lc{\svv}$, where $i$ is 
the appex number of $v$, such that $\rvv$ is a 
first vertex from $v$ to every vertex in $B_r=]\mvv,\svv]_\ri{L}$ and vertex $\lvv$ is a 
first vertex from $v$ to every vertex in $B_l=]\svv,\lvv]_\ri{L}$. 
Since Theorem \ref{separator} implies that the vertices in $B_v$ can be distributed 
by assigning $B_l$ to directed edge $(v,\lvv)$ and $B_r$ to directed edge $(v,\rvv)$, the main theorem 
(Theorem \ref{mainTheo}) 
follows for the case in which neither dominating vertices nor counter vertices exist.




\begin{theorem}\label{separator}
Let $(\mc{X},\ri{C})$ be a clique-cycle for a strict circular-arc graph $G=(V,E)$ without 
dominating vertices and without counter vertices, $\ri{L}$ given by 
Listing \ref{fig:animals}, $v\in V$,  $\lvv$ the left 
vertex of $v$, $\mvv$ the middle vertex of $v$, and $\rvv$ the right vertex of $v$.
There exists a vertex $\svv\in~]\mvv,\lvv[_\ri{L}$ such that, for 
every vertex $w\in~]\mvv,\svv]_\ri{L}$, we have $\rvv\in\fe{v}{w}$, and 
for every vertex $w'\in~]\svv,\lvv]_\ri{L}$, 
we have $\lvv\in\fe{v}{w'}$.
\end{theorem}

\begin{proof}
Let $i$ the apex number of $v$.

If $i=1$ the theorem holds for $\svv=\ri{L}^{-1}(\lvv)$, since every vertex in 
$]\mvv,\lvv[_\ri{L}$ is adjacent to $\rvv$.

If $i>1$ and $\lvi{v}{i}=\rvi{v}{i}$, we have
$$\forall w\in~]\mvv, \last{\rc{\rvi{v}{i-1}}}]_\ri{L}\colon\rvd{w}\leq i-1\leq\lvd{w}$$
and
$$\forall w\in~]\last{\rc{\rvi{v}{i-1}}},\lvv[_\ri{L}\colon\lvd{w}\leq\rvd{w},$$
wherefore, by Corollary \ref{b}, Theorem \ref{separator} holds for $\svv=\last{\rc{\rvi{v}{i-1}}}$.

The following argumentation is partially illustrated in Fig. \ref{neu}.
\begin{figure}[t]
\begin{center}
\includegraphics[width=0.9\textwidth]{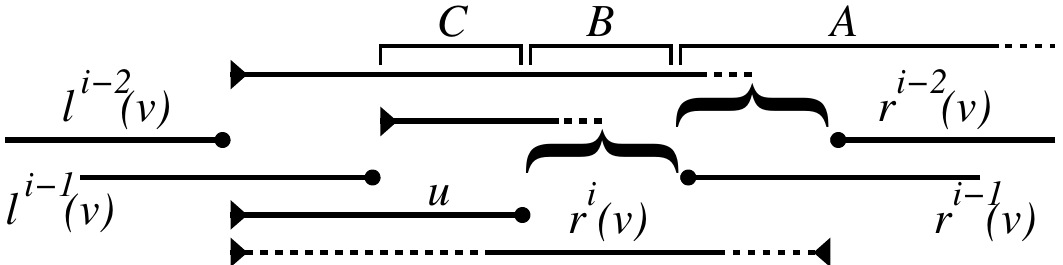}
\end{center}
\caption{
This picture illustrates the proof of Theorem \ref{separator} and partially shows a 
circular-arc graph, where the lines represent arcs.
A line ends with a triangle, if it may reach not further to left or right, respectively.
The upper two lines with a dotted end illustrate that an arc, that has its right end in 
the interval depicted by the rotated 
curly bracket, may reach no further to the left than the triangle of the respective line.
We begin to search for $u$ at the right end of the lower curly bracket.
The rotated square brackets on top mark the areas of the left cliques of 
the vertices in sets $A,B,C$.
The case in which $\lvi{v}{i-1}= u$ looks similar.
}
\label{neu}
\end{figure}

If $i>1$ and $\lvi{v}{i}\neq\rvi{v}{i}$, let
$$A=]\mvv,\last{\rc{r^{i-1}(v)}}]_\ri{L},$$
and $a\in A$.
We have $\rvd{a}\leq i-1$, because $a$ is adjacent to a vertex 
in $\{\rvi{v}{1},\rvi{v}{2},\ldots,\rvi{v}{i-1}\}$.
When $\lvv$ reaches further to the right than $v$, $a$ might appear in a clique in $]\rc{v},\rc{\lvv}]_\ri{C}$ and, therefore, be adjacent to $\lvv$, which implies $\lvd{a} = 1$.
However, in this case, we also have $\rvd{a} = 1$. 
Therefore, we could only have $\lvd{a} <  \rvd{a}$, when $a$ reaches so far to the right, such that it is adjacent to $\lvi{v}{i-2}$.
However, in this case $a$ would reach further to the left than $\lvi{v}{i-1}$, because $a$ is adjacent to $\rvi{v}{i-1}$ and $\lvi{v}{i-1}$ is \q{at most} adjacent to $\rvi{v}{i}$.
Because $\lvi{v}{i-1}$ is the left vertex of $\lvi{v}{i-2}$ this contradicts Theorem \ref{no_vertex_further_than_lv}.
%
Thus, Corollary \ref{b} implies 
\begin{equation}\label{eqd}
\forall a\in A\colon\rvv\in \fe{v}{a}.
\end{equation}

Now we go through $\ri{L}$ beginning at 
$\ri{L}(\last{\rc{\rvi{v}{i-1}}}=\first{\ri{C}(\rc{\rvi{v}{i-1}})}$,\footnote{This 
equality holds by Corollary \ref{coro:tail vertexheal}.} in other 
words, we begin at the first vertex after the \q{last} vertex in  $A$, until we 
eventually come across a vertex $u$ that either is $\lvi{v}{i-1}$ or adjacent 
to $\lvi{v}{i-1}$.\footnote{This vertex 
$u$ may be $l^i(v)$. Also there may be vertices in $A$ that are adjacent to 
$\lvi{v}{i-1}$, in other words, we may have \q{passed by} a vertex that is 
adjacent to $\lvi{v}{i-1}$ already.}
We define
$$B=
]\last{\rc{\rvi{v}{i-1}}},u[_\ri{L}
$$
i.e. $B$ is the set of vertices after $\last{\rc{\rvi{v}{i-1}}}$ and before $u$ in $\ri{L}$.
Note that, by the choice of $u$, no vertex in $B$ is adjacent to $\lvi{v}{i-1}$.

Assume there is a vertex $x\in B$ that is not adjacent to $\rvi{v}{i}$.
When we go through $\ri{C}$ beginning at $\lc{\rvi{v}{i-1}}$ we eventually come across $\rc{\rvi{v}{i-1}}$ 
(since $\rvi{v}{i}$ and $\rvi{v}{i-1}$ are adjacent, we have $\lc{\rvi{v}{i}}\in[\lc{\rvi{v}{i-1}},\rc{\rvi{v}{i-1}}]_\ri{C}$).
Before we come across $\lc{x}$, we come across $\rc{\rvi{v}{i}}$, since otherwise $x$ and $\rvi{v}{i}$ would be adjacent.
We do not come  across $\lc{\lvi{v}{i-1}}$ until we came across $\rc{x}$, since otherwise $x$ and $\lvi{v}{i-1}$ would be adjacent.
It follows that $\lvi{v}{i-1}$ and $\rvi{v}{i}$ are not adjacent and, because $\lvi{v}{i}$ is adjacent to  $\lvi{v}{i-1}$ and $\rvi{v}{i}$, it follows that $\lvi{v}{i}$ reaches further to the right than $\rvi{v}{i}$.
Because $\lvi{v}{i}$ and $\rvi{v}{i}$ are adjacent, $\lc{\lvi{v}{i}}$ 
must appear in $[\lc{\rvi{v}{i}},\rc{\rvi{v}{i}}]_\ri{C}$.
It cannot appear in $[\lc{\rvi{v}{i}},\rc{\rvi{v}{i-1}}]_\ri{C}$, because then $\lvi{v}{i}$ 
is adjacent to $\rvi{v}{i-1}$,
which contradicts Definition \ref{de_r_v}, since $\lvi{v}{i}$ reaches further to the right 
than $\rvi{v}{i}$ and $\rvi{v}{i}$ is the right vertex of $\rvi{v}{i-1}$.
Therefore, we have
$\lc{\lvi{v}{i}}\in~]\rc{\rvi{v}{i-1}},\rc{\rvi{v}{i}}]_\ri{C}$.
Because we came across $\rc{\rvi{v}{i}}$ before $\lc{x}$, we have $]\rc{\rvi{v}{i-1}},\rc{\rvi{v}{i}}]_\ri{C}\subset~]\rc{\rvi{v}{i-1}},\lc{x}[_\ri{C}$ and, therefore,
\begin{equation}\label{wtf}
\lc{\lvi{v}{i}}\in~]\rc{\rvi{v}{i-1}},\lc{x}[_\ri{C}.
\end{equation}
Since the vertices in $\ri{L}$ are primary ordered by their left clique, Inclusion (\ref{wtf}) implies
\begin{equation}\label{wtfwtf}
\lvi{v}{i}\in[\first{\ri{C}(\rc{\rvi{v}{i-1}})},x]_\ri{L}.
\end{equation}
Remember that we began at $\first{\ri{C}(\rc{\rvi{v}{i-1}})}$ 
to go through $\ri{L}$ until we found a vertex $u$ that is adjacent to $\lvi{v}{i-1}$.
Inclusion \ssref{wtfwtf} implies that we must  have come across $\lvi{v}{i}$ before $x$.
But then we have $u=\lvi{v}{i}$.
Since
$$x\notin[\first{\ri{C}(\rc{\rvi{v}{i-1}})},\lvi{v}{i}[_\ri{L}=[\ri{L}(\last{\rc{\rvi{v}{i-1}}}),u[_\ri{L}=B$$ 
contradicts the initial assumption, every vertex in $B$ is adjacent to $\rvi{v}{i}$.

Since by the choice of $u$ no vertex in $B$ is adjacent to $\lvi{v}{i-1}$, we have
$$\forall b\in B\colon\lvd{b}\geq i=\rvd{b},$$
which, by Corollary \ref{b}, implies
\begin{equation}\label{eqc}
\forall b\in B\colon\rvv\in \fe{v}{b}.
\end{equation}
We define $\svv=\ri{L}^{-1}(u)$.
Since the ring-intervals $A$ and $B$ appear consecutively in $\ri{L}$, 
by Construction \ref{ringkombi}, we have
\begin{equation}\label{bbt}
A\cup B=~]\mvv,u[_\ri{L}=~]\mvv,\svv]_\ri{L}.
\end{equation}
From \ssref{eqd}, \ssref{eqc}, and \ssref{bbt} it follows
\begin{equation}\label{eqw1}
\forall w\in]\mvv,\svv]\colon\rv{v}\in \fe{v}{w}.
\end{equation}

Since the vertices in $\ri{L}$ are primary ordered by their left cliques, 
no vertex $$w\in[\first{\ri{C}(\rc{\rvi{v}{i-1}})},\lvi{v}{i-1}]_\ri{L}$$  is adjacent 
to $\rvi{v}{i-1}$ and therefore we have $\rvd{w}\geq i$. We define
$$C=]\svv,\lvi{v}{i-1}[_\ri{L}$$
i.e. $C$ is the set of vertices after $\svv$ and before $\lvi{v}{i-1}$ in $\ri{L}$.
Note that $$B,C\subseteq[\first{\ri{C}(\rc{\rvi{v}{i-1}})},\lvi{v}{i-1}]_\ri{L}.$$
We have
\begin{equation}\label{inca}
\forall \mc{c}\in[\lc{\ri{L}(\svv)},\lc{\lvi{v}{i-1}}]_\ri{C}\colon\lvi{v}{i}\in \mc{c}
\end{equation}
since otherwise $u=\ri{L}(\svv)$ would reach further to the left than $\lvi{v}{i}$ although both are 
adjacent to $\lvi{v}{i-1}$ and $\lvi{v}{i}$ is the left vertex of $\lvi{v}{i-1}$.
This would contradict Theorem \ref{no_vertex_further_than_lv}.
It follows
$$\forall c\in C\colon\lvd{c}\leq i\leq \rvd{c},$$
which, by Corollary \ref{b}, implies 
\begin{equation}\label{eqa}
\forall c\in C\colon\lvv\in \fe{v}{c}.
\end{equation}
Obviously we have 
\begin{equation}\label{eqb}
\forall w\in[\lvi{v}{i-1},\lvv[_\ri{L}\colon\lvv\in \fe{v}{w}.
\end{equation}
By Construction \ref{ringkombi}, we have 
$C\cup[\lvi{v}{i-1},\lvv[_\ri{L}=]\svv,\lvv[_\ri{L}$
and therefore \ssref{eqa} and \ssref{eqb} imply 
\begin{equation}\label{eqw2}
\forall w\in]\svv,\lvv[_\ri{L}\colon \lvv\in \fe{v}{w}.
\end{equation}
The statement follows from \ssref{eqw1} and \ssref{eqw2}.
\end{proof}

In Section \ref{lvbisv} we assigned $[\lvv,x]_\ri{L}$ to directed edge $(v,\lvv)$
for some vertex $x$ to the left of $v$.
If $]\svv,\lvv[_\ri{L}$ is assigned to directed edge $(v,\lvv)$ 
additionally (as implied by Theorem \ref{separator}) 
the intervals on directed edge $(v,\lvv)$ can be compressed.
Therefore, if neither dominating vertices nor counter vertices exist, 
$v$ given $\ri{L}$ suffices a shortest path 2-SIRS, 
but only $(v,\rvv)$ is assigned two ring-intervals.

Below Definition \ref{de_r_v} we stated that $v$ given $\ri{L}$ suffices 
a shortest path 1-SIRS, if we can choose $\rvv=\lvv$ or $\rvv=\mvv$.
This is the case because then the ring-intervals assigned to directed edge $(v,\rvv)$ 
can be compressed as well.

Theorem \ref{theo:1sirs} covers a further case where $v$ given $\ri{L}$ suffices a shortest path 
1-SIRS.
Since Theorem \ref{theo:1sirs} is not essential for the proof of Theorem \ref{mainTheo}, 
we leave the proof to the interested reader.

\begin{theorem}\label{theo:1sirs}
Let $(\mc{X},\ri{C})$ be a clique-cycle for a circular-arc graph $G=(V,E)$, $\ri{L}$ 
given by Listing \ref{fig:animals}, $v\in V$ not a dominating vertex, $\lvv$ the 
left vertex of $v$, and $\mvv$ the middle vertex of $v$.
If $\mvv$ is adjacent to every vertex face-to-face with $v$ that $\lvv$ is not adjacent 
to, $v$ suffices a shortest path 1-SIRS.
\end{theorem}

\paragraph{Dominating Vertices or Counter Vertices exist}\label{gyros}

In the foregoing section the existence of dominating and counter vertices was excluded and is therefore discussed in this section.
The vertices to be distributed are those face-to-face with $v$, i.e., the vertices in $B_v$.
By Theorem \ref{Lvrl_v}, the only vertices in $B_v$ that are adjacent to $v$ are dominating vertices;
this case is considered first.
Next, the case where dominating vertices not face-to-face with $v$ or counter vertices of $v$ exist is addressed. 
Finally, the case where a pair of counter vertices exists (but $v$ has no counter vertex) is discussed. 

\subparagraph{Dominating Face-to-Face Vertices exist}

Listing \ref{fig:animals} ensures that all dominating vertices appear consecutively in $\ri{L}$, 
wherefore we can find two dominating vertices $d_l$ and $d_r$ such that $[d_l,d_r]_{ \ri{L}}=\dom$. 
Let $x\in V$ such that $]\mvv, x]_{ \ri{L}} = B_v$.
Since by Theorem \ref{Lvrl_v}, $\dom$ are the only vertices in $B_v$ that are adjacent to $v$, we can assign $]\mvv,d_l]_{ \ri{L}}$ to directed edge $(v,d_l)$, $[d_r,x]_{ \ri{L}}$ to directed edge $(v,d_r)$, and for the remaining dominating vertices $d_i\in~]d_l,d_r[$, $[d_i,d_i]_\ri{L}$ to directed edge $(v,d_i)$.
It follows that $v$ given $\ri{L}$ suffices a shortest path 1-SIRS.

\subparagraph{Dominating Face-to-Face Vertices  do not exist}

If there are no dominating vertices in $B_v$, Theorem \ref{Lvrl_v} implies that no vertex in $B_v$ is adjacent to $v$.
Therefore, by Corollary \ref{adjacentToVorCoV}, every vertex that
is a counter vertex of $v$ is adjacent to every vertex in $B_v$.
%
%
Let $u\in\co{v}\cup\dom$ and $w\in B_v$.
Then $u$ is a first vertex 
from $v$ to $w$, wherefore we can assign ring-interval $B_v$ to 
directed edge $(v,u)$, which had hitherto one ring-interval assigned.
Thus, $v$ given $\ri{L}$ suffices a shortest path 2-SIRS, wherein only directed edge $(v,u)$ has two ring-intervals assigned.
By Theorem \ref{theo:1sirs}, $v$ given $\ri{L}$ suffices a 
shortest path 1-SIRS, if $\mvv\in\co{v}\cup\dom$.

\subparagraph{Pairs of Counter Vertices Exist but $\mathbf{v}$ is in none of them}\label{counteratall}

Assume $G$ does not contain dominating 
vertices and the considered vertex $v$  has no counter vertices but 
there exists at least one pair $(w,c_w)$ of counter vertices.
We have $B_v=~]\mvv,\lvv[_\ri{L}$ and 
no vertex in $B_v$ is adjacent to $v$.
We make a distinction of cases for whether $v$ is adjacent to one or to both counter vertices.

Assume $v$ is adjacent to $w$ as well as $c_w$.
There is a clique $\mc{c}\in\mc{X}$ that contains $v$, $w$, and $c_w$.
When we begin at $\mc{c}$ to go through $\ri{C}$,
we eventually reach a clique $\mc{d}$ that contains only one of $w$ or $c_w$
and not the other.
Without loss of generality, let $w$ be contained in $\mc{d}$ and let $\mc{e}$ be the first clique that contains $c_w$ again.
$\mc{e}$ is the left clique of $c_w$ and 
contains $\lvv$ (we might have $c_w=\lvv$), as
otherwise, $c_w$ would reach further to the left than $\lvv$, 
which contradicts Theorem \ref{no_vertex_further_than_lv}.
Since $w$ and $c_w$ are counter vertices, all cliques we came across contained $w$.
Since the vertices in $\ri{L}$ are primary ordered by their 
left cliques and we also must have traversed $\rc{v}$, which contains $\mvv$, every vertex in 
$B_v=~]\mvv,\lvv[_\ri{L}$ appeared in at least one of the cliques we came across.
Therefore, $w$ is adjacent to every vertex in $B_v$.
Since $\ema{I}_v$ hitherto maps directed edge $(v,w)$ to one ring-interval, we can assign $B_v$ to directed edge $(v,w)$. 
Thus, $v$ given $\ri{L}$ suffices a shortest path 2-SIRS, wherein only directed edge $(v,w)$ is 
actually mapped to two ring-intervals.

Now assume $v$ is adjacent to one of the counter vertices and not to the other.
Without loss of generality, let this vertex be $w$.
Let $r$ be a vertex adjacent to $v$ that, similar to Definition \ref{de_r_v}, reaches farthest to the right.
Since, by Theorem \ref{no_vertex_further_than_lv}, $\lvv$ is one of the vertices 
adjacent to $v$ that reaches farthest to the left, $\lvv$ as well as $r$ are adjacent to $c_w$.
In fact, we might have $w=\lvv$ or $w=r$.
Let $\mv{r}$ be the middle vertex of $r$ and $A_r=~]r,\mv{r}]_\ri{L}$ the set of vertices to the right of $r$.
$A_r$ is not empty, since otherwise $G$ would be an interval graph.
Also $\mvv$ is contained in $A_r$, if $r\neq\mvv$.
By Theorem \ref{Lvrl_v}, every vertex in $A_r$ is adjacent to $r$.
Consider we go through $\ri{L}$ beginning at $\ri{L}(\mvv)$ until we reach $\mv{r}$ or $\lvv$.
If we have $\lvv=\mv{r}$, we consider it as reaching $\lvv$.
If we first reach $\lvv$, we exactly traversed the vertices in $]\mvv,\lvv]_\ri{L}\subseteq B_v$.
It follows that we have $B_v\subseteq A_r$ and thus
assign $B_v=]\mvv,\lvv[_\ri{L}$ to directed edge $(v,r)$.
%
Now assume we first reach $\mv{r}$. 
By Theorem \ref{Lvrl_v}, every vertex in $]\mvv,\mv{r}]_\ri{L}\subseteq A_r$ 
is adjacent to $r$, wherefore we can assign $]\mvv,\mv{r}]_\ri{L}$ to directed edge $(v,r)$.
If $\ri{L}(\mv{r})=\lvv$, all vertices in $B_v$ are distributed.
Otherwise let $X=]\mv{r},\lvv[_\ri{L}\subset B_v$ be the vertices left to be distributed.
Since, by Theorem \ref{no_vertex_further_than_lv}, $\lvv$ is a vertex adjacent to $v$ 
that reaches farthest to the left and $\lvv$ is adjacent to $c_w$, 
for every vertex $u\in X$ either $(v,\lvv,u)$ or $(v,\lvv,c_w,u)$ is a shortest path.
Thus, we can assign $X$ to directed edge $(v,\lvv)$ (and compress the two intervals on this edge),
wherefore all vertices are distributed and only directed edge $(v,r)$ is assigned two ring intervals.

\section{Improving the Space Requirement}\label{improving}

This section shows that although circular-arc graphs do not allow shortest path 1-IRSs in general, 
the number of intervals in a shortest path 2-SIRS for a circular-arc graph is 
bounded close to the number of intervals in an 1-IRS.
Let $G_d=(V,A)$ be the directed symmetric version of a strict circular-arc graph.
A shortest path 2-SIRS for $G_d$ can be obtained by ordering the vertices of $G_d$ 
according to Listing \ref{fig:animals}
and then labeling the outgoing directed edges of every vertex according to 
Section \ref{lvbisv} and \ref{mvbislv}.
As pointed out in these sections, there is at most one outgoing edge per vertex that is assigned two intervals.
We have $|V|\leq|A|/2$, since $G$ is a strict circular-arc graph, and $|V|=|A|/2$ only holds when $G$ is a ring in which case our labeling yields a shortest path 1-SIRS.
Therefore, the number of intervals in the constructed 2-SIRS is less than $|A|+|V| = 1.5\cdot |A|$, while already a 1-IRS for $G_d$ permits up to $|A|$ intervals.

\section{Implementation}\label{implement}

This section outlines how the interval routing scheme implied by the proof 
of Theorem \ref{mainTheo} can be implemented.
Let $G=(V,E)$ be a circular-arc graph with $|V|=n$ and $|E|=m$ and 
$(\mc{X},\ri{C})$ be a clique-cycle for $G$ with $|\mc{X}|=c$
maximal cliques ($\mc{X}$ does not need to contain sets of vertices; plain 
elements representing the cliques are sufficient).
We discuss how the vertex order $\ri{L}$ and the critical vertices 
$\lvv,\mvv,\rvv,\svv$ for every vertex $v\in V$ can be determined in $\ema{O}(c+n^2)$ time.
We assume that every vertex $v\in V$ has a pointer $\lc{v}$ to 
its left clique and a pointer $\rc{v}$ to its right clique in $\mc{X}$.

We define \emph{broadness of a vertex $v$} as the number of cliques 
in $\mc{X}$ that contain $v$.
By numbering the elements in $\mc{X}$ as they appear in $\ri{C}$, beginning at 
an arbitrary element, we construct a bijective mapping 
$\fu{N}\colon\mc{X}\rightarrow\{1,2,\ldots,|\mc{X}|\}$, which allows to compute the 
broadness of a vertex in  $\mathcal{O}(1)$ time.
Sorting the vertices ascending by their broadness 
requires $\mathcal{O}(n\cdot\log(n))$ time. We introduce a list $\fu{Q}(\mc{c})$ for 
every clique $\mc{c}\in\mc{X}$. Next the vertices are traversed in the sorted order and 
every vertex $v$ is added to list $\fu{Q}(\lc{v})$.
By concatenating the lists as their respective cliques appear in $\ri{C}$ we obtain the vertex order
$\ri{L}$ in $\mathcal{O}(n\cdot\log(n))$ time.
The left vertex $\lvv$ of a given vertex $v\in V$ can be determined as follows:
we inspect all vertices that are adjacent to $v$ and store the vertex $w$ that reaches 
farthest to the left (for two adjacent vertices $w,w'$, we can determine in constant time 
which one reaches further to the left, by considering $\fu{N}(\lc{w})$ and $\fu{N}(\lc{w'})$) 
and appears first in $\fu{Q}(\lc{w})$ (this can also be decided in constant time, by numbering 
the elements in each list consecutively). Vertex $w$ must be $\lvv$.
It follows that the set of left vertices can be computed in $\mathcal{O}(m)$ time.
The set of right vertices can be computed in similarly.
Determining the set of middle vertices is simpler and can be accomplished in 
$\mathcal{O}(n)$.
For a given vertex $v$, $\lvi{v}{i}$ and $\rvi{v}{i}$ can be determined by traversing 
the sequences $\lvi{v}{1},\lvi{v}{2},\ldots$ and $\rvi{v}{1},\rvi{v}{2},\ldots$ synchronously; 
then the determination of  $\svv$ is a simple task and can be done in $\fu{O}(n)$
for a specific vertex and in $\fu{O}(n^2)$ for all vertices.

Since the algorithm includes constant time operations on the elements in $\mc{X}$ 
and $|\mc{X}|=c$ is linear within the number of vertices, the overall runtime 
is $\mathcal{O}(c+n^2+m)=\mathcal{O}(n^2)$.
The algorithm can be extended to  handle the existence of dominating 
and counter vertices without exceeding this time bound.


\section{Conclusions}\label{conc}

We showed that the class of circular-arc graphs has strict compactness 2.
Throughout the proof special cases of strict circular-arc graphs that allow 
1-SIRSs were highlighted (not all cases were pointed out for the sake of brevity), which is in particular interesting, since the class of circular-arc graphs is a super-class of the class of interval and unit circular-arc graphs, which always allow 1-SIRSs.
We showed
that the constructed 2-SIRS requires a number of intervals that is 
closer to the maximal number of intervals in an 1-IRS than in a 2-IRS. 
An open question is, if strict circular-arc graphs exist that allow 1-SIRS but the vertex ordering generated by
Listing \ref{pseudoLabeling} does not allow these.



\bibliographystyle{alpha}
\bibliography{bibfile}


\end{document}